\documentclass[preprint,10pt]{style_perso}
\usepackage{amsmath,amsthm,amsfonts}
  \usepackage{paralist}
  \usepackage{graphics} 
  \usepackage{epsfig} 

\usepackage{graphicx}  
\usepackage{epstopdf}
 \usepackage[colorlinks=true]{hyperref}
\hypersetup{urlcolor=blue, citecolor=red}

\newtheorem{theorem}{Theorem}[section]
\newtheorem{corollary}{Corollary}

\newtheorem{lemma}[theorem]{Lemma}
\newtheorem{proposition}{Proposition}

\theoremstyle{definition}
\newtheorem{definition}[theorem]{Definition}
\newtheorem{remark}{Remark}
\newtheorem{assumption}{Assumption}

\DeclareMathAlphabet{\mathpzc}{OT1}{pzc}{m}{it}
\mathchardef\mhyphen="2D

\begin{document}
\begin{frontmatter}



\title{Harmonic analysis and a bentness-like notion in certain finite  Abelian groups over some finite fields}


\author{Laurent Poinsot}
\ead{laurent.poinsot@lipn.univ-paris13.fr} 
\address{Universit\'e Paris 13, Paris Sorbonne Cit\'e, 
   LIPN, CNRS (UMR 7030), France}


\begin{abstract}
It is well-known that degree two finite field extensions can be equipped with a Hermitian-like structure similar to the extension of the complex field over the reals. In this contribution, using this structure, we develop a modular character theory and the appropriate Fourier transform for some particular kind of finite  Abelian groups. Moreover we introduce the notion of bent functions for finite field valued functions rather than usual complex-valued functions, and we study several of their properties. In particular we prove that this bentness notion is a consequence of that of Logachev, Salnikov and Yashchenko, 
introduced in \emph{Bent functions on a finite  Abelian group} (1997). In addition this new bentness notion is also generalized to a vectorial setting. 
\end{abstract}

\end{frontmatter}
\section{Introduction}
\label{sec1}
The most simple Hermitian structure is given by the complex field $\mathbb{C}$ when equipped with the complex modulus $\overline{z}$ (for $z \in \mathbb{C}$). Although quite simple, this structure has many applications in the theory of harmonic analysis of finite  Abelian groups. Indeed the theory of characters for such groups is explicitly based on the existence of a special subgroup of the multiplicative group $\mathbb{C}^*$, the {\it unit sphere} or {\it group of roots of unity} $\mathcal{S}(\mathbb{C}) = \{\, z \in \mathbb{C} \colon \overline{z}z = 1\,\}$. This multiplicative group contains an isomorphic copy of each possible cyclic group. Thus it is possible to represent an abstract group $G$ as a group $\widehat{G}$ of $\mathcal{S}(\mathbb{C})$-valued functions that preserves the group structure of $G$, called {\it characters}, which is isomorphic to $G$ itself. These characters are the group homomorphisms from $G$ to $\mathcal{S}(\mathbb{C})$. Moreover the dual group $\widehat{G}$ is also an orthogonal basis for the $|G|$-th dimensional vector space of $\mathbb{C}$-valued functions defined on $G$. This property makes it possible to carry out a harmonic analysis on finite  Abelian groups using the (discrete) Fourier transform which is defined as the decomposition of a vector of $\mathbb{C}^G$ in the basis of characters.

Given a degree two extension $\textup{GF}(p^{2n})$ of $\textup{GF}(p^n)$, the Galois field with $p^n$ elements where $p$ is a prime number, we can also define a ``conjugate'' and thus a Hermitian structure on $\textup{GF}(p^{2n})$ in a way similar to the relation $\mathbb{C} / \mathbb{R}$. In particular this makes possible the definition of a unit circle $\mathcal{S}(\textup{GF}(p^{2n}))$ which is a cyclic group of order $p^n + 1$, subgroup of the multiplicative group $\textup{GF}(p^{2n})^*$ of invertible elements. The analogy with $\mathbb{C} / \mathbb{R}$ is extended in this paper by the definition of $\textup{GF}(p^{2n})$-valued characters of finite  Abelian groups $G$ as group homomorphisms from $G$ to $\mathcal{S}(\textup{GF}(p^{2n}))$. But $\mathcal{S}(\textup{GF}(p^{2n}))$ does obviously not contain a copy of each cyclic group. Nevertheless if $d$ divides $p^n + 1$, then the cyclic group $\mathbb{Z}_d$ of modulo $d$ integers  embeds as a subgroup of this particular unit circle. It forces our modular theory of characters to be applied only to direct products of the form $G = \displaystyle \prod_{i=1}^N \mathbb{Z}_{d_i}^{m_i}$ where each $d_i$ divides $p^n + 1$. In addition we prove that these modular characters form an orthogonal basis (by respect to the Hermitian-like structure $\textup{GF}(p^{2n})$ over $\textup{GF}(p^n)$). This decisive property makes it possible the definition of an appropriate notion of Fourier transform for $\textup{GF}(p^{2n})$-valued functions, rather than $\mathbb{C}$-valued ones, defined on $G$,  as their decompositions in the dual basis of characters. In this contribution we largely investigate several properties of this modular version of the Fourier transform similar to classical ones.

Traditionally an important cryptographic criterion can be naturally defined in terms of Fourier transform. Indeed {\it bent functions} are those functions $f : G \rightarrow \mathcal{S}(\mathbb{C})$ such that the magnitude of their Fourier transform $|\widehat{f}(\alpha)|^2$ is constant, equals to $|G|$. Such functions achieve the optimal possible resistance against the famous linear cryptanalysis of secret-key cryptosystems. Now using our theory of characters we can translate the bentness concept in our modular setting in order to treat the case of $\mathcal{S}(\textup{GF}(p^{2n}))$-valued functions defined on a finite  Abelian group $G$. In this paper are also studied some properties of such functions. As a last contribution, we develop a vectorial notion of bent functions that concerns maps from $G$ to $\textup{GF}(p^{2n})^l$ that explicitly uses an Hermitian structure of $\textup{GF}(p^{2n})^l$. 

We warm the reader that the new notion of bentness presented hereafter is introduced as an illustration of this new finite-field valued character theory and its associated Fourier transform. The possible connections between usual bent functions (in particular those with values in a finite-field, see~\cite{Ambrosimov}) and our own definition are not all made clear in the present contribution. This paper should only be seen  as a complete presentation of a general framework about modular harmonic analysis, given by a modular character theory and an associated modular Fourier transform, that should possibly be used for future research  to make interesting connections with cryptographic Boolean functions.  In this contribution we limit ourselves to point out that the objects introduced hereafter share many properties with their well-known counterparts.  Deeper relationships, if they exist, are outside the scope of our current work. Nevertheless we mention the important assertion proved in this contribution: many usual bent functions in the sense of Logachev, Salnikov and Yashchenko (see~\cite{LSY}) are also bent functions in our finite-field setting.

\section{Character theory: the classical approach}
\label{sec2}
In this paper $G$ always denotes a finite Abelian group (in additive representation), $0_G$ is its identity element.  Moreover for all groups $H$, $H^*$ is the set obtained from $H$ by removing its identity element (therefore, $G^*=G \setminus \{\,0_G\,\}$). This last notation is in accordance with the usual notation $\mathbb{N}^*=\mathbb{N}\setminus\{\, 0\,\}$. 

The character theory of finite Abelian groups was originally introduced in order to embed algebraic structures into the complex field $\mathbb{C}$, and therefore to obtain geometric realizations of abstract groups as sets of complex transformations. The main relevant objects are the {\it characters}, {\it i.e.} the group homomorphisms from a finite Abelian group $G$ to the unit circle $\mathcal{S}(\mathbb{C})$ of the complex field. The set of all such characters of $G$ together with point-wise multiplication is denoted by $\widehat{G}$ and called the {\it dual group of} $G$. A classical result claims that $G$ and its dual are isomorphic This property essentially holds because $\mathcal{S}(\mathbb{C})$ contains an isomorphic copy of all cyclic groups. Usually the image in $\widehat{G}$ of $\alpha \in G$ by such an isomorphism is denoted by $\chi_{\alpha}$. The complex vector space $\mathbb{C}^G$ of complex-valued functions defined on $G$ can be equipped with an inner product defined for $f,g\in \mathbb{C}^G$ by
\begin{equation}
\langle f,g \rangle = \displaystyle \sum_{x\in G}f(x)\overline{g(x)}
\end{equation} where $\overline{z}$ denotes the complex conjugate of $z \in \mathbb{C}$. With respect to this Hermitian structure, $\widehat{G}$ is an orthogonal basis, {\it i.e.}
\begin{equation}
\langle \chi_{\alpha},\chi_{\beta} \rangle = \left \{
\begin{array}{l l}
0 & \mbox{if}\ \alpha \not = \beta,\\
|G|& \mbox{if}\ \alpha = \beta
\end{array} \right .
\end{equation}
for $\alpha,\beta \in G^2$. We observe that in particular (replacing $\beta$ by $0_G$),
\begin{equation}
\displaystyle \sum_{x\in G}\chi_{\alpha}(x) = \left \{
\begin{array}{l l}
0& \mbox{if}\ \alpha \not = 0_G,\\
|G|& \mbox{if}\ \alpha = 0_G.
\end{array}\right .
\end{equation}
According to the orthogonality property, the notion of characters leads to some harmonic analysis of finite Abelian groups {\it via} a (discrete) Fourier transform.
\begin{definition}
Let $G$ be a finite Abelian group and $f : G \rightarrow \mathbb{C}$. The (discrete) Fourier transform of $f$ is defined as
\begin{equation}
\begin{array}{l l l l}
\widehat{f}\colon & G& \rightarrow & \mathbb{C}\\
& \alpha & \mapsto & \displaystyle \sum_{x \in G}f(x)\chi_{\alpha}(x)\ .
\end{array}
\end{equation}
\end{definition}
The Fourier transform of a function $f$ is essentially its decomposition in the basis $\widehat{G}$. This transform is invertible and one has an {\it inversion formula} for $f$,
\begin{equation}
\begin{array}{l l l l}
f(x) = \displaystyle \frac{1}{|G|}\sum_{\alpha \in G} \widehat{f}(\alpha)\overline{\chi_{\alpha}(x)}
\end{array}
\end{equation}
for each $x \in G$. More precisely the Fourier transform is an algebra isomorphism from $(\mathbb{C}^G,*)$ to $(\mathbb{C}^G,.)$ where the symbol ``$.$'' denotes the point-wise multiplication of functions, while $*$ is the convolutional product defined for $f,g \in (\mathbb{C}^G)$ by
\begin{equation}
\begin{array}{l l l l}
f*g\colon & G & \rightarrow & \mathbb{C}\\
& \alpha & \mapsto & \displaystyle  \sum_{x \in G}f(x)g(-x+\alpha)
\end{array}
\end{equation}
Since the Fourier transform is an isomorphism between the two algebras, the {\it trivialization of the convolutional product} holds for each $(f,g) \in (\mathbb{C}^{G})^2$ and each $\alpha \in G$,\begin{equation}
\widehat{(f*g)}(\alpha) = \widehat{f}(\alpha)\widehat{g}(\alpha)\ .
\end{equation}
From these two main properties one can establish the following classical results.
\begin{proposition}
Let $G$ be a finite Abelian group and $f,g \in \mathbb{C}^G$. We have
\begin{equation}
\displaystyle \sum_{x \in G}f(x)\overline{g(x)} = \frac{1}{|G|}\sum_{\alpha \in G}\widehat{f}(\alpha)\overline{\widehat{g}(\alpha)}\ \mbox{{\it (Plancherel formula),}}
\end{equation}
\begin{equation}
\displaystyle \sum_{x \in G}|f(x)|^2 = \frac{1}{|G|}\sum_{\alpha \in G}|\widehat{f}(\alpha)|^2\ \mbox{{\it (Parseval equation)}}
\end{equation}
 where $|z|$ is the complex modulus of $z \in \mathbb{C}$.
\end{proposition} 
This paper is mostly dedicated to the study of a character theory for some finite Abelian groups over some finite fields rather than $\mathbb{C}$. In particular we provide similar results as those from this section. Obviously we need an Hermitian structure over the chosen finite field. This is the content of the next section. 

\section{Hermitian structure over finite fields}
\label{sec3}
In this section we recall some results about an Hermitian structure in some kinds of finite fields. By analogy with the classical theory of characters (recalled in section~\ref{sec2}), this particular structure is involved in the definition of a suitable theory of finite field-valued characters of some finite  Abelian groups which is developed in section~\ref{sec4}. This section is directly inspired from~\cite{NihoBent} of which we follow the notations, and generalized to any characteristic $p$.

Let $p$ be a given prime number and $q$ an even power of $p$, {\it i.e.}, there is $n \in {\mathbb{N}}^*$ such that $q = p^{2n}$, and in particular $q$ is a square. 
\begin{assumption}
From now on the parameters $p,n,q$ are fixed as introduced above.
\end{assumption}
As usually $\textup{GF}(q)$ is the finite field of characteristic $p$ with $q$ elements and by construction $\textup{GF}(\sqrt{q})$ is a subfield of $\textup{GF}(q)$.  The field $\textup{GF}(q)$, as an extension of degree $2$ of $\textup{GF}(\sqrt{q})$, is also a vector space of dimension $2$ over $\textup{GF}(\sqrt{q})$. This situation is similar to the one of $\mathbb{C}$ and $\mathbb{R}$. As $\textup{GF}(q)$ plays the role of $\mathbb{C}$, the Hermitian structure should be provided for it. Again according to the analogy $\mathbb{C}/\mathbb{R}$, we then need to determine a corresponding conjugate. In order to do this we use the {\it Frobenius automorphism} $\textup{Frob}$ of $\textup{GF}(q)$ defined by
\begin{equation}
\begin{array}{l l l l}
\textup{Frob} : & \textup{GF}(q) & \rightarrow & \textup{GF}(q)\\
& x & \mapsto & x^{p}
\end{array}
\end{equation}
and one of its powers
\begin{equation}
\begin{array}{l l l l}
\textup{Frob}_k : & \textup{GF}(q) & \rightarrow & \textup{GF}(q)\\
& x & \mapsto & x^{p^k}\ .
\end{array}
\end{equation}
In particular $\textup{Frob}_1 = \textup{Frob}$.
\begin{definition}
The {\it conjugate of} $x \in \textup{GF}(q)$ {\it over} $\textup{GF}(\sqrt{q})$ is denoted by $\overline{x}$ and defined as 
\begin{equation}
\overline{x} = \textup{Frob}_{n}(x) = x^{p^n} = x^{\sqrt{q}}\ .
\end{equation} 
\end{definition}
\noindent In particular, for every $n\in\mathbb{Z}$, $\overline{n1}=n1$. 
The field extension $\textup{GF}(q)/\textup{GF}(\sqrt{q})$ has amazing similarities with the extension $\mathbb{C}$ over the real numbers in particular regarding the conjugate.
\begin{proposition}
Let $x_1,x_2 \in \textup{GF}(q)^2$, then
\begin{enumerate}
\item $\overline{x_1 + x_2} = \overline{x_1} + \overline{x_2}$,
\item $\overline{-x_1} = -\overline{x_1}$,
\item $\overline{x_1x_2} = \overline{x_1}\ \overline{x_2}$,
\item $\overline{\overline{x_1}} = x_1$.
\end{enumerate}
\end{proposition}
\begin{proof}
The three first points come from the fact that $\textup{Frob}_{n}$ is a field homomorphism of $\textup{GF}(q)$. The last point holds since for each $x \in \textup{GF}(q)$, $x^{q} = x$.
\end{proof}
The {\it relative norm with respect to} $\textup{GF}(q)/\textup{GF}(\sqrt{q})$ is defined as
\begin{equation}
\textup{norm} (x) = x\overline{x}
\end{equation}
for $x \in \textup{GF}(q)$, and it maps $\textup{GF}(q)$ to $\textup{GF}(\sqrt{q})$. We observe that $\textup{norm}(x) \in \textup{GF}(\sqrt{q})$ because $\sqrt{q} + 1$ divides $q - 1$, and $\textup{norm}(x)=0$ if, and only if, $x=0$.\\

The {\it unit circle} of $\textup{GF}(q)$ is defined as the set
\begin{equation}
\mathcal{S}(\textup{GF}(q)) = \{\, x \in \textup{GF}(q) \colon x \overline{x} = 1\,\}
\end{equation}
of all elements having relative norm $1$. By construction $\mathcal{S}(\textup{GF}(q))$ is the group of $(\sqrt{q} + 1)$-th roots of unity, and therefore it is a (multiplicative) cyclic group of order $\sqrt{q} + 1$ since $\textup{GF}(q)^*$ is cyclic and $\sqrt{q} + 1$ divides $q-1$. 
In what follows, $\mathcal{S}(\textup{GF}(q))$ will play exactly the same role as $\mathcal{S}(\mathbb{C})$ in the classical theory of characters.

Now let $\textup{GF}(q)^l$ be the $l$-dimensional vector space over $\textup{GF}(q)$, then the {\it Hermitian dot product} of two vectors $x=(x_1,\ldots,x_l)$ and $y=(y_1,\ldots,y_l)$ of $\textup{GF}(q)^l$ is
\begin{equation}\label{Hdp}
\langle x,y\rangle = \displaystyle \sum_{i=1}^l x_i \overline{y_i}\ .
\end{equation}
Again, this kind of Hermitian dot product has properties similar to the natural Hermitian inner product on complex vector spaces. Let $\alpha,\beta \in \textup{GF}(q)$ and $x,y,z \in \textup{GF}(q)^l$, then
\begin{enumerate}
\item $\langle (\alpha x + \beta y),z\rangle = \alpha \langle x,z \rangle + \beta \langle y,z\rangle$ ({\it linearity}),
\item $\langle x,y \rangle = \overline{\langle y,x\rangle}$ (conjugate symmetry),
\item $\langle x,x\rangle \in \textup{GF}(\sqrt{q})$.
\end{enumerate}
We observe that the canonical basis $B$ of $\textup{GF}(q)^{l}$ over $\textup{GF}(q)$ is \emph{orthonormal} with respect to $\langle\cdot,\cdot\rangle$ ($\langle e,e^{\prime}\rangle=0$ if $e\not=e^{\prime}$, and $\langle e,e\rangle=1$ for all $e,e^{\prime}\in B$), and it is clear that for every $x\in \textup{GF}(q)^l$, $x=\sum_{e\in B}\langle x,e\rangle e$. Nevertheless, contrary to the usual Hermitian situation, it may happen that for a non-zero vector $x\in \textup{GF}(q)^l$, $\langle x,x\rangle=0$ (for instance, consider the situation where $q=2^{2n}$, and $l=2m$). But this dot product is \emph{non-degenerate}: let us assume that for some $x$, $\langle x,y \rangle=0$ for every $y$, then $x=(\underbrace{0,\ldots,0}_{l\ \mathit{factors}})$ (to see this it suffices to let $y$ run over the canonical basis of $\textup{GF}(q)^{l}$). Similarly, by conjugate symmetry, $\langle x,y\rangle=0$ for each $x$ implies that $y=(\underbrace{0,\ldots,0}_{l\ \mathit{factors}})$. Therefore, $\langle \cdot,\cdot\rangle$ defines a \emph{pairing} (see~\cite{Boneh}).  

We denote $\textup{norm}_l (x) = \langle x,x \rangle = \displaystyle \sum_{i=1}^{l}\textup{norm}(x)$ for $x \in \textup{GF}(q)^l$, and finally $\mathcal{S}(\textup{GF}(q)^l)$ is defined as the hypersphere in $\textup{GF}(q)^l$ with center at $(\underbrace{0,\ldots,0}_{l\ \mathit{factors}})$ and radius $1$.
\section{Characters of certain finite  Abelian groups over a finite field}
\label{sec4}
Before beginning some formal developments, one should warn the reader on the limitations of the expected character theory in finite fields. In section~\ref{sec3}, we claimed that $\mathcal{S}(\textup{GF}(q))$ is a cyclic group of order $\sqrt{q} + 1$. Then for each nonzero integer $d$ that divides $\sqrt{q} + 1$, there is a (cyclic) subgroup of $\mathcal{S}(\textup{GF}(q))$ of order $d$,  and this is the unique kind of subgroups. As a character theory is essentially used to faithfully represent an abstract group as an isomorphic group of functions, a copy of such group must be contained in the corresponding unit circle. Then our character theory in $\textup{GF}(q)$ will only apply on groups for which all their factors in a representation as a product direct group of cyclic subgroups divides $\sqrt{q} + 1$.
\begin{assumption}
From now on $d$  always denotes an element of ${\mathbb{N}}^*$ that divides $\sqrt{q} + 1$.
\end{assumption} 
\begin{definition}(and proposition)
The (cyclic) subgroup of $\mathcal{S}(\textup{GF}(q))$ of order $d$ is denoted by $\mathcal{S}_d (\textup{GF}(q))$. In particular, $\mathcal{S}(\textup{GF}(q)) = \mathcal{S}_{\sqrt{q} + 1}(\textup{GF}(q))$. If $u$ is a generator of $\mathcal{S}(\textup{GF}(q))$ then $u^{\frac{\sqrt{q} + 1}{d}}$ is a generator of $\mathcal{S}_d (\textup{GF}(q))$.
\end{definition}
A {\it character} of a finite  Abelian group $G$ {\it with respect to} $\textup{GF}(q)$ (or simply a {\it character}) is a group homomorphism from $G$ to $\mathcal{S}(\textup{GF}(q))$. Since a character $\chi$ is $\mathcal{S}(\textup{GF}(q))$-valued, $\chi(-x) = (\chi(x))^{-1} = \overline{\chi(x)}$, $\textup{norm}(\chi(x)) = 1$ and $\chi(0_G) = 1$ for each $x \in G$. By analogy with the traditional version, we denote by $\widehat{G}$ the set of all characters of $G$ that we call its {\it dual}. When equipped with the point-wise multiplication, $\widehat{G}$ is a finite  Abelian group. One recall that this multiplication is defined as
\begin{equation}
\forall \chi,\chi^{\prime} \in \widehat{G},\ \chi\chi^{\prime}\colon x\mapsto \chi(x) \chi^{\prime}(x)\ .
\end{equation}
As already mentionned in introduction, we focus on a very special kind of finite  Abelian groups: the additive group of modulo $d$ integers $\mathbb{Z}_d$ which is identified with the subset $\{0,\ldots,d-1\}$ of $\mathbb{Z}$.
\begin{theorem}\label{cyclic-case}
The groups $\mathbb{Z}_d$ and $\widehat{\mathbb{Z}_d}$ are isomorphic.
\end{theorem}
\begin{proof}
The parameter $d$ has been chosen so that it divides $\sqrt{q} + 1$. Then there is a unique (cyclic) subgroup $\mathcal{S}_d(\textup{GF}(q))$ of $\mathcal{S}(\textup{GF}(q))$ of order $d$. Let $u_d$ be a generator of this group. Then the elements of $\widehat{\mathbb{Z}_d}$ have the form, for $j \in \mathbb{Z}_d$,
\begin{equation}
\chi_j \colon\left\{
\begin{array}{l l l}
\mathbb{Z}_d &\rightarrow & \mathcal{S}_d(\textup{GF}(q))\\
k & \mapsto & (u_d^j)^k = u_d^{jk}.
\end{array}
\right .
\end{equation}
Actually the characters are $\mathcal{S}_d(\textup{GF}(q))$-valued since for each $x \in \mathbb{Z}_d$ and each character $\chi$, $\chi(x) \in \mathcal{S}(\textup{GF}(q))$ by definition, and satisfies $1 = \chi(0) = \chi(dx) = (\chi(x))^d$ and then $\chi(x)$ is a $d$-th root of the unity. Then to determine a character $\chi \in \widehat{\mathbb{Z}_d}$, we need to compute the value of $\chi(k) = \chi(k1)$ for $k \in \{0,\ldots,d-1\}$, which gives
\begin{equation}
\chi(k) = u_d^{jk}\ .
\end{equation}
In this equality, we have denoted $\chi(1)$ by $u_d^j$ for $j \in \{0,\ldots,d-1\}$ since $\chi(1)$ is a $d$-th root of the unity in $\mathcal{S}(\textup{GF}(q))$. Then the character $\chi$ belongs to $\{\,\chi_0,\ldots,\chi_{d-1}\,\}$. Conversely, we observe that for $j \in \{1,\ldots,d-1\}$, the maps $\chi_j$ are group homomorphisms from $\mathbb{Z}_d$ to $\mathcal{S}(\textup{GF}(q))$ so they are elements of $\widehat{\mathbb{Z}_d}$. Let us define the following function.
\begin{equation}
\begin{array}{l l l l}
\Psi:& \mathbb{Z}_d &\rightarrow& \widehat{\mathbb{Z}_d}\\
& j & \mapsto & \chi_j\ .
\end{array}
\end{equation}
We have already seen that it is onto. Moreover, it is also one-to-one (it is sufficient to evaluate $\chi_j = \Psi(j)$ at $1$) and it is obviously a group homomorphism. It is then an isomorphism, so that $\widehat{\mathbb{Z}_d}$ is isomorphic to $\mathbb{Z}_d$. 
\end{proof}
The isomorphism established in theorem~\ref{cyclic-case} between a group and its dual can be generalized as follows.
\begin{proposition}\label{product-case}
$\mathbb{Z}_{d_1} \times \mathbb{Z}_{d_2}$ and $\widehat{(\mathbb{Z}_{d_1}\times \mathbb{Z}_{d_2})}$ are isomorphic.
\end{proposition}
\begin{proof}
The proof is easy since it is sufficient to remark that $\widehat{(\mathbb{Z}_{d_1}\times \mathbb{Z}_{d_2})}$ and $\widehat{\mathbb{Z}_{d_1}}\times \widehat{\mathbb{Z}_{d_2}}$ are isomorphic. We recall that $d_1$ and $d_2$ are both assumed to divide $\sqrt{q} + 1$, thus $\widehat{\mathbb{Z}_{d_1}}$ and $\widehat{\mathbb{Z}_{d_2}}$ exist and are isomorphic to $\mathbb{Z}_{d_1}$ and $\mathbb{Z}_{d_2}$ respectively. Let $i_1$ be the first canonical injection of $\mathbb{Z}_{d_1} \times \mathbb{Z}_{d_2}$ and $i_2$ the second (when $\mathbb{Z}_{d_1} \times \mathbb{Z}_{d_2}$ is seen as a direct sum). The following map
\begin{equation}
\Phi \colon \left\{
\begin{array}{l l l}
\widehat{(\mathbb{Z}_{d_1}\times \mathbb{Z}_{d_2})} & \rightarrow & \widehat{\mathbb{Z}_{d_1}}\times \widehat{\mathbb{Z}_{d_2}}\\
\chi & \mapsto & (\chi \circ i_1,\chi \circ i_2)
\end{array}
\right .
\end{equation}
is a group isomorphism. It is obviously one-to-one and for $(\chi_1,\chi_2) \in \widehat{\mathbb{Z}_{d_1}}\times \widehat{\mathbb{Z}_{d_2}}$, the map $\chi : (x_1,x_2) \mapsto \chi_1(x_1) \chi_2(x_2)$ is an element of $\widehat{(\mathbb{Z}_{d_1}\times \mathbb{Z}_{d_2})}$ and $\Phi(\chi) = (\chi_1,\chi_2)$. Then $\widehat{(\mathbb{Z}_{d_1}\times \mathbb{Z}_{d_2})}$ is isomorphic to $\mathbb{Z}_{d_1} \times \mathbb{Z}_{d_2}$ since $\widehat{Z_{d_i}}$ and $\mathbb{Z}_{d_i}$ are isomorphic (for $i=1,2$).
\end{proof}
From proposition~\ref{product-case} it follows in particular that $\widehat{\mathbb{Z}_{d}^m}$ is isomorphic to $\mathbb{Z}_{d}^m$. 
This result also provides a specific form to the characters of $\mathbb{Z}_d^m$ as follows. We define a dot product, which is a $\mathbb{Z}_d$-bilinear map from $(\mathbb{Z}_d^m)^2$ to $\mathbb{Z}_d$, by 
\begin{equation}
x\cdot y = \displaystyle \sum_{i=1}^{m}x_i y_i \in \mathbb{Z}_d
\end{equation}
for $x,y\in \mathbb{Z}_d^m$. Then the character that corresponds to $\alpha \in \mathbb{Z}_d^m$ can be defined by
\begin{equation}
\begin{array}{l l l l}
\chi_{\alpha} \colon& \mathbb{Z}_d^m & \rightarrow &\mathcal{S}_d(\textup{GF}(q))\\
& x & \mapsto & u_d^{\alpha\cdot x}
\end{array}
\end{equation}
where $u_d$ is a generator of $\mathcal{S}_d(\textup{GF}(q))$. In particular for each $\alpha,x \in \mathbb{Z}_d^m$, $\chi_{\alpha}(x) = \chi_x (\alpha)$. The following result is obvious. 
\begin{corollary}\label{coreasy}
Let $G \cong \displaystyle \prod_{i=1}^N \mathbb{Z}_{d_i}^{m_i}$ be a finite  Abelian group for which each integer $d_i$ divides $\sqrt{q}+1$. Then $G$ and $\widehat{G}$ are isomorphic.
\end{corollary}
\begin{remark}
The fact that $G\cong \widehat{G}$ does not depend on a decomposition of $G$ into a direct sum of cyclic groups. But a particular isomorphism of corollary~\ref{coreasy} depends on the decomposition $\prod_{i=1}^N \mathbb{Z}_{d_i}^{m_i}$ of the group $G$.
\end{remark}
If $G = \displaystyle \prod_{i=1}^N \mathbb{Z}_{d_i}^{m_i}$ satisfies the assumption of the  corollary~\ref{coreasy}, then we can also obtain a specific form for its characters and a specific isomorphism from $G$ to its dual. Let $\alpha = (\alpha_1,\ldots,\alpha_N) \in G$. 
\begin{equation}\label{specific-isomorphism}
\begin{array}{l l l l}
\chi_{\alpha} \colon& G & \rightarrow & \mathcal{S}(\textup{GF}(q))\\
& x=(x_1,\ldots,x_N) & \mapsto & \displaystyle \prod_{i=1}^N u_{d_i}^{\alpha_i \cdot x_i}
\end{array}
\end{equation} 
where for each $i \in \{1,\ldots,N\}$, $u_{d_i}$ is a generator of $\mathcal{S}_{d_i}(\textup{GF}(q))$. In particular for each $\alpha,x \in G^2$, we also have $\chi_{\alpha}(x) = \chi_x (\alpha)$.
\begin{assumption}\label{division-of-the-order}
From now on, each finite  Abelian group $G$ considered is assumed to be of a specific form $\displaystyle \prod_{i=1}^N \mathbb{Z}_{d_i}^{m_i}$ where for each $i \in \{1,\ldots,N\}$, $d_i$ divides $\sqrt{q} + 1$, so that we have at our disposal a specific isomorphism given by the formula~(\ref{specific-isomorphism}) between $G$ and $\widehat{G}$.
\end{assumption}

The dual $\widehat{G}$ of $G$ is constructed and is shown to be isomorphic to $G$. We may also be interested into the \emph{bidual} $\widehat{\widehat{G}}$ of $G$, namely the dual of $\widehat{G}$. Similarly to the usual situation of complex-valued characters, we prove that $G$ and its bidual are canonically isomorphic. It is already clear that $G\cong \widehat{\widehat{G}}$ (because $G\cong\widehat{G}$ and $\widehat{G}\cong \widehat{\widehat{G}}$). But this isomorphism is far from being canonical since it depends on a decomposition of $G$, and of $\widehat{G}$, and choices for generators of each cyclic factor in the given decomposition. We observe that the map $e\colon G\rightarrow \widehat{\widehat{G}}$ defined by $e(x)(\chi)=\chi(x)$ for every $x\in G$, $\chi\in\widehat{G}$ is a group homomorphism.  To prove that it is an isomorphism it suffices to check that $e$ is one-to-one (since $G$ and $\widehat{\widehat{G}}$ have the same order). Let $x\in\ker(e)$. Then, for all $\chi\in\widehat{G}$, $\chi(x)=1$. Let us fix an isomorphism $\alpha\in G\rightarrow \chi_{\alpha}\in \widehat{G}$ as in the formula~(\ref{specific-isomorphism}). Then, for every $\alpha\in G$, $\chi_{\alpha}(x)=1=\chi_{x}(\alpha)$ so that $x=0_G$. Thus we have obtained an appropriate version of Pontryagin-van Kampen duality (see~\cite{Hewitt}). Let us recall that according to the \emph{structure theorem of finite Abelian groups}, for any finite Abelian group $G$, there is a unique finite sequence of positive integers, called the \emph{invariants of $G$},  $d_1,\cdots,d_{\ell_G}$ such that $d_i$ divides $d_{i+1}$ for each $i<\ell_G$.  Let us denote by $\mathpzc{Ab}_{\sqrt{q}+1}$ the category of all finite Abelian groups $G$ such that $d_{\ell_G}$ divides $\sqrt{q}+1$, with usual homomorphisms of groups as arrows. From the previous results, if $G$ is an object of $\mathpzc{Ab}_{\sqrt{q}+1}$, then $G\cong \widehat{G}$. Moreover, $\widehat{(\cdot)}$ defines a contravariant functor (see~\cite{McLane}) from $\mathpzc{Ab}_{\sqrt{q}+1}$ to itself. Indeed, if $\phi\colon G\rightarrow H$ is a homomorphism of groups (where $G,H$ belongs to $\mathpzc{Ab}_{\sqrt{q}+1}$), then $\widehat{\phi}\colon \widehat{H}\rightarrow \widehat{G}$ defined by $\widehat{\phi}(\chi)=\chi\circ \phi$ is a homomorphism of groups. Then, we have the following duality theorem.
\begin{theorem}[Duality]
The covariant (endo-)functor $\widehat{\widehat{(\cdot)}}\colon \mathpzc{Ab}_{\sqrt{q}+1}\rightarrow \mathpzc{Ab}_{\sqrt{q}+1}$ is a (functorial) isomorphism (this means in particular that $G\cong \widehat{\widehat{G}}$).
\end{theorem}

\section{Orthogonality relations}
The characters satisfy a certain kind of orthogonality relation. In order to establish it we introduce the natural ``action'' of $\mathbb{Z}$ on any finite field $\textup{GF}(p^l)$ of characteristic $p$ as $kx = x \underbrace{+\ldots+}_{k\ \mathit{times}} x$ for $(k,x) \in \mathbb{Z} \times \textup{GF}(p^l)$. This is nothing else than the fact that the underlying Abelian group structure of $\textup{GF}(p^l)$  is a $\mathbb{Z}$-module. In particular one has for each $(k,k',x) \in \mathbb{Z} \times \mathbb{Z} \times \textup{GF}(p^l)$,
\begin{enumerate}
\item $0x = 0$, $1x = x$ and $k0 = 0$,
\item $(k+k')x = kx + k^{\prime}x$ and then $nkx = n(kx)$,
\item $k1 \in \textup{GF}(p)$, $k1 = (k\bmod p)1$, $k^m1 = (k1)^m$ and if $k \bmod p \not = 0$, then $(k1)^{-1} = (k \bmod p)^{-1}1$.
\end{enumerate}
In the remainder we identify $k1$ with $k \bmod p$ or in other terms we make an explicit identification of $\textup{GF}(p)$ by $\mathbb{Z}_p$.
\begin{lemma}\label{lem_sum_car_0_ou_card_G}
Let $G$ be a finite  Abelian group. For $\chi \in \widehat{G}$,
\begin{equation}
\displaystyle \sum_{x \in G} \chi(x) = \left \{
\begin{array}{l l}
0 & \mbox{if}\ \chi \not = 1\ ,\\
(|G| \bmod p) &\mbox{if}\ \chi = 1\ .
\end{array}
\right .
\end{equation}
\end{lemma}
\begin{proof}
If $\chi = 1$, then $\displaystyle \sum_{x \in G}1 = (|G|\bmod p)$ since the characteristic of $\textup{GF}(q)$ is equal to $p$. Let us suppose that $\chi \not =1$. Let $x_0 \in G$ such that $\chi(x_0) \not = 1$. Then we have 
\begin{equation}
\displaystyle \chi(x_0)\sum_{x \in G}\chi(x) = \sum_{x \in G}\chi(x_0 + x) = \sum_{y \in G}\chi(y),
\end{equation}
so that $\displaystyle (\chi(x_0) - 1)\sum_{x \in G}\chi(x) = 0$ and thus $\displaystyle \sum_{x \in G}\chi(x) = 0$ (because $\chi(x_0) \not = 1$). 
\end{proof}
This technical lemma allows us to define the orthogonality relation between characters.
\begin{definition}
Let $G$ be a finite  Abelian group. Let $f,g \in \textup{GF}(q)^G$. We define the ``inner product'' of $f$ and $g$ by
\begin{equation}
\langle f,g\rangle = \displaystyle \sum_{x \in G}f(x)\overline{g(x)} \in \textup{GF}(q).
\end{equation}
\end{definition}
\noindent The above definition does not ensure that $\langle f,f\rangle=0$ implies that $f\equiv 0$ as it holds for a true inner product. Indeed, take $q=2^{2n}$, and let $f\colon \mathbb{Z}_2\rightarrow \textup{GF}(2^{2n})$ be the constant map with value $1$. Then, $\langle f,f\rangle=0$. Thus, contrary to a usual Hermitian dot product, an orthogonal family (with respect to $\langle\cdot,\cdot\rangle$) of $\textup{GF}(q)^G$ is not necessarily $\textup{GF}(q)$-linearly independent. 
\begin{proposition}[Orthogonality relation]
Let $G$ be a finite  Abelian group. For all $(\chi_1,\chi_2) \in \widehat{G}^2$ then
\begin{equation}
\langle \chi_1,\chi_2 \rangle = \left \{
\begin{array}{l l}
0 & \mbox{if}\ \chi_1 \not = \chi_2,\\
|G|\bmod p & \mbox{if}\ \chi_1 = \chi_2.
\end{array}
\right .
\end{equation}
\end{proposition}
\begin{proof}
Let us denote $\chi = \chi_1 \chi_2^{-1} = \chi_1 \overline{\chi_2}$. We have:
\begin{equation}
\langle \chi_1,\chi_2 \rangle = \sum_{x \in G}\chi(x).
\end{equation}
If $\chi_1 = \chi_2$, then $\chi = 1$ and if $\chi_1 \not = \chi_2$, then $\chi \not = 1$. The proof is obtained by using the previous lemma~\ref{lem_sum_car_0_ou_card_G}.  
\end{proof}
\begin{remark}
The term {\it orthogonality} would be abusive if $|G| \bmod p = 0$, because then  $\displaystyle \sum_{x \in G}\chi(x) = 0$ for all $\chi \in \widehat{G}$. Nevertheless we know from the assumption~\ref{division-of-the-order} that all the $d_i$'s divide $\sqrt{q}+1=p^n+1$. In particular, $d_i=1\bmod p$ and therefore $|G|=\prod_i d_i^{m_i}$ is co-prime to $p$, and the above situation cannot occur, so $|G|$ is invertible modulo $p$. 
\end{remark}

\section{Fourier transform over a finite field}\label{sec5}

In this section is developed a Fourier transform for functions defined on $G$ and based on the theory of characters introduced in section~\ref{sec4}. There is already a Fourier transform with values in some finite field called {\it Mattson-Solomon transform}~\cite{Wolfmann} but it maps a function $f \in \textsc{GF}(q)^{\mathbb{Z}_d}$ to a function $g \in \textsc{GF}(q^m)^{\mathbb{Z}_d}$ where $m$ is the smallest positive integer so that $d$ divides $q^m - 1$. In this paper we want our Fourier transform to ``live'' in a finite field $\textsc{GF}(q)$ and not in one of its extensions. Moreover the existing transform is not based on an explicit Hermitian structure nor on a theory of characters.  For these reasons, we need to introduce a new kind of Fourier transform.

Let $u$ be a generator of $\mathcal{S}(\textup{GF}(q))$. Let $G$ be a finite  Abelian group and $f : G \rightarrow \textup{GF}(q)$. We define the following function.
\begin{equation}\label{ff-fourier}
\begin{array}{l l l l}
\widehat{f}\colon& \widehat{G}& \longrightarrow & \textup{GF}(q)\\
& \chi & \mapsto & \displaystyle \sum_{x \in G} f(x)\chi(x)\ .
\end{array}
\end{equation}
\begin{remark}
We warm the reader that we use the same notation $\widehat{f}$ as the one used for the Fourier transform of a complex-valued function.  From now on only the second definition is used.
\end{remark}
Because $G = \displaystyle \prod_{i=1}^{N}\mathbb{Z}_{d_i}^{m_i}$, by using the isomorphism between $G$ and its dual group from section~\ref{sec4}, we can define
\begin{equation}
\begin{array}{l l l l}
\widehat{f}:& G& \longrightarrow & \textup{GF}(q)\\
& \alpha & \mapsto & \displaystyle \displaystyle \sum_{x \in G} f(x)\chi_{\alpha}(x) = \sum_{x \in G} f(x)\displaystyle \prod_{i=1}^N u^{\frac{(\sqrt{q} + 1)\alpha_i\cdot x_i}{d_i}}
\end{array}
\end{equation} 
Let us compute $\widehat{\widehat{f}}$. Let $\alpha \in G$. We have
\begin{equation}\label{double_FT}
\begin{array}{l l l}
\widehat{\widehat{f}}(\alpha) &=& \displaystyle \sum_{x \in G}\widehat{f}(x)\chi_{\alpha}(x)\\
&=& \displaystyle \sum_{x \in G} \sum_{y \in G}f(y)\chi_x(y)\chi_{\alpha}(x)\\
&=&\displaystyle \sum_{x \in G} \sum_{y \in G}f(y)\chi_y(x)\chi_{\alpha}(x)\\
&=&\displaystyle \sum_{y \in G}f(y) \sum_{x \in G}\chi_{\alpha + y}(x)\\
&=&(|G|\bmod p) f(-\alpha)
\end{array}
\end{equation}
The last equality holds since $$\displaystyle \sum_{x \in G}\chi_{\alpha + y}(x) = \left \{
\begin{array}{l l}
0 & \mbox{if}\ y \not = -\alpha\ ,\\
(|G|\bmod p) & \mbox{if}\ y = -\alpha\ .
\end{array} \right .$$
Now if we  assume that $(|G|\bmod p)=0$, then it follows that the function $f\mapsto \widehat{f}$ is non invertible but this situation cannot occur since from the assumption~\ref{division-of-the-order}, $|G|$ is invertible modulo $p$. 
Therefore we can claim that the function $\widehat{(\cdot)}$ that maps $f \in \textup{GF}(q)^{G}$ to $\widehat{f} \in \textup{GF}(q)^{G}$ is invertible. It is referred to as the {\it Fourier transform} of $f$ (with respect to $\textup{GF}(q)$) and it admits an {\it inversion formula}: for $f \in \textup{GF}(q)^{G}$ and for each $x \in G$,
\begin{equation}
f(x) = \displaystyle (|G| \bmod p)^{-1} \sum_{\alpha \in G}\widehat{f}(\alpha)\overline{\chi_{\alpha}(x)}
\end{equation}
where $(|G| \bmod p)^{-1}$ is the multiplicative inverse of $(|G| \bmod p)$ in $\mathbb{Z}_p$ (this inverse exists according to the choice of $G$). This Fourier transform shares many properties with the classical discrete Fourier transform.
\begin{definition}
Let $G$ be a finite  Abelian group. Let $f,g \in \textup{GF}(q)^G$. For each $\alpha \in G$, we define the {\it convolutional product} of $f$ and $g$ at $\alpha$ by
\begin{equation}
(f*g)(\alpha) = \displaystyle \sum_{x \in G}f(x)g(-x + \alpha)\ . 
\end{equation}
\end{definition}
\begin{proposition}[Trivialization of the convolutional product] Let $f,g \in \textup{GF}(q)^{G}$. For each $\alpha \in G$, 
\begin{equation}
\widehat{(f*g)}(\alpha) = \widehat{f}(\alpha)\widehat{g}(\alpha)\ .
\end{equation}
\end{proposition}
\begin{proof}
Let $\alpha \in G$. We have
\begin{equation}
\begin{array}{l l l}
\widehat{(f*g)}(\alpha) &=& \displaystyle \sum_{x \in G}(f*g)(x)\chi_{\alpha}(x)\\
&=& \displaystyle \sum_{x \in G}\sum_{y \in  G}f(y)g(-y+x)\chi_{\alpha}(x)\\
&=& \displaystyle \sum_{x \in G}\sum_{y \in G}f(y)g(-y+x)\chi_{\alpha}(y-y+x)\\
&=& \displaystyle \sum_{x \in G}\sum_{y \in G}f(y)g(-y+x)\chi_{\alpha}(y)\chi_{\alpha}(-y+x)\\
&=& \widehat{f}(\alpha)\widehat{g}(\alpha)\ .
\end{array}
\end{equation}
\end{proof}
\noindent We recall that the group-algebra $\textup{GF}(q)[G]$ of $G$ over $\textup{GF}(q)$ is the $\textup{GF}(q)$-vector space $\textup{GF}(q)^G$ with point-wise addition, and with the convolution product.  We observe that the Fourier transform $\widehat{(\cdot)}$ is an algebra isomorphism from the group-algebra  $\textup{GF}(q)[G]$ of $G$  its usual convolution product to $\textup{GF}(q)[G]$ with the point-wise product. Moreover, let $(\delta_x)_{x\in G}$ be the canonical basis of $\textup{GF}(q)^G$ (as a $\textup{GF}(q)$-vector space), that is, $\delta_x(y)=0$ if $x\not=y$ and $\delta_x(x)=1$. It is easy to see (using a fixed isomorphism between $G$ and $\widehat{G}$) that $\widehat{\delta}_x=\chi_x$. Because $\widehat{(\cdot)}$ is an isomorphism, this means that $(\chi_x)_{x\in G}$ is a basis of $\textup{GF}(q)^G$ over $\textup{GF}(q)$, and it turns that the Fourier transform $\widehat{f}$ of $f\in \textup{GF}(q)^G$ is the decomposition of $f$ into the basis of characters (we recall here that the fact for a family of elements of $\textup{GF}(q)^G$ to be orthogonal with respect to the inner-product $\langle\cdot,\cdot\rangle$ of $\textup{GF}(q)^G$ does not ensure that the family into consideration is linearly independent because $\langle\cdot,\cdot\rangle$ is not positive-definite). 

We continue to enunciate formulas obtained by considering the decomposition into the basis of characters. 
\begin{proposition}[Plancherel formula] Let  $f,g \in \textup{GF}(q)^{G}$. Then,
\begin{equation}
\displaystyle \sum_{x \in G}f(x)\overline{g(x)} = (|G| \bmod p)^{-1} \sum_{\alpha \in G}\widehat{f}(\alpha)\overline{\widehat{g}(\alpha)}\ .
\end{equation}
\end{proposition}
\begin{proof} Let us define the following functions for any finite group $G$ and any function $h : G \rightarrow \textup{GF}(q)$,
\begin{equation}
\begin{array}{l l l l}
I_G \colon& G & \rightarrow & G\\
& x & \mapsto & -x\\
&\mbox{and}&\\
\overline{h} \colon& G & \rightarrow & \textup{GF}(q)\\
& x & \mapsto & \overline{h(x)}\ .
\end{array}
\end{equation}
Then we have $(f * \overline{g}\circ I_{G})(0_G) = \displaystyle \sum_{x \in G} f(x)\overline{g(x)}$. But from the inversion formula we have also
\begin{equation}
\begin{array}{l l l}
(f*\overline{g}\circ I_{G})(0_G) &=& \displaystyle (|G| \bmod p)^{-1} \sum_{\alpha \in G}\widehat{(f * \overline{g}\circ I_{G})}(\alpha)\\
&=& \displaystyle (|G| \bmod p)^{-1} \sum_{\alpha \in G}\widehat{f}(\alpha)\widehat{(\overline{g} \circ I_{G})}(\alpha)\\
&&\mbox{(by the trivialization of the convolutional product.)}
\end{array}
\end{equation}
Let us compute $\widehat{(\overline{g} \circ I_{G})}(\alpha)$ for $\alpha \in G$.
\begin{equation}
\begin{array}{l l l}
\widehat{(\overline{g} \circ I_{G})}(\alpha)&=&\displaystyle \sum_{x \in G}(\overline{g} \circ I_{G})(x)\chi_{\alpha}(x)\\
&=&\displaystyle \sum_{x \in G}\overline{g(-x)}\chi_{\alpha}(x)\\
&=&\displaystyle \sum_{x \in G}\overline{g(x)}\chi_{\alpha}(-x)\\
&=&\displaystyle \sum_{x \in G}\overline{g(x)}(\chi_{\alpha}(x))^{-1}\\
&=&\displaystyle \sum_{x \in G}\overline{g(x)}\overline{\chi_{\alpha}(x)}\\
&=& \displaystyle \overline{\sum_{x \in G}g(x)\chi_{\alpha}(x)}\\
&=&\overline{\widehat{g}(\alpha)}\ .
\end{array}
\end{equation}
Then we obtain the equality
\begin{equation}
(f * \overline{g}\circ I_{G}) = \displaystyle (|G| \bmod p)^{-1} \sum_{\alpha \in G}\widehat{f}(\alpha)\overline{\widehat{g}(\alpha)}
\end{equation}
that ensures the correct result. 
\end{proof}
\begin{corollary}[Parseval equation] Let  $f,g \in \textup{GF}(q)^{G}$. Then
\begin{equation}
\displaystyle \sum_{x \in G}\textup{norm}(f(x)) = (|G| \bmod p)^{-1}\sum_{\alpha \in G} \textup{norm}(\widehat{f}(\alpha))\ .
\end{equation}
In particular, if $f$ is $\mathcal{S}(\textup{GF}(q))$-valued, then
\begin{equation}
\displaystyle \sum_{\alpha \in G} \textup{norm}(\widehat{f}(\alpha)) = (|G| \bmod p)^2\ .
\end{equation}
\end{corollary}
\begin{proof}
It is sufficient to apply Plancherel formula with $g=f$.   
\end{proof}

\section{Bent functions over a finite field}
\label{sec6}
Up to now, the following ingredients have been introduced: an Hermitian-like structure on degree two extensions, a finite-field character theory for finite Abelian groups (of order co-prime to the characteristic), and a corresponding Fourier transform. All of them may be constituents of a bentness-like notion in this particular setting. As already explained in introduction, although we are aware of an existing notion of bent functions in finite fields~\cite{Ambrosimov}, in this contribution we do not make any interesting connections with these maps and those introduced below, except that they share very similar properties. Nevertheless, we compare the notion of bentess due to Logachev, Salnikov and Yashchenko in~\cite{LSY} to our, and we prove that many bent functions as defined in~\cite{LSY} are also bent functions in our setting.  The notion of bentness introduced now serves also as an illustration of our character theory. In this section we also prove the existence of functions which are bent in a sense presented hereafter.

In the traditional setting, {\it i.e.}, for complex-valued functions defined on any finite  Abelian group $G$, bent functions \cite{Carlet,Dill,LSY,Nyberg,Roth} are those maps $f \colon G \rightarrow \mathcal{S}(\mathbb{C})$ such that for each $\alpha \in G$,
\begin{equation}
|\widehat{f}(\alpha)|^2 = |G|\ .
\end{equation}
This notion is closely related to some famous cryptanalysis namely the differential~\cite{Bih} and linear~\cite{Mat} attacks on secret-key cryptosystems. We translate this concept in the current finite-field setting as follows.
\begin{definition}\label{new-bent}
The map $f \colon G \rightarrow \mathcal{S}(\textup{GF}(q))$ is called bent if for all $\alpha \in G$,
\begin{equation}
\textup{norm}(\widehat{f}(\alpha)) = (|G| \bmod p)\ .
\end{equation}
\end{definition} 
\subsection{Derivative and bentness} 
In the traditional approach the relation with bentness and  differential attack is due to the following result.
\begin{proposition}\cite{LSY}
Let $f \colon  G \rightarrow \mathcal{S}(\mathbb{C})$. The function $f$ is bent if, and only if, for all $\alpha \in G^*$,
\begin{equation}
\displaystyle \sum_{x \in G}f(\alpha + x)\overline{f(x)} = 0\ .
\end{equation}
\end{proposition}
Similarly, it is possible to characterize the new concept of bentness in a similar way. Let $f \colon G \rightarrow \textup{GF}(q)$. For each $\alpha \in G$, we define the {\it derivative of} $f$ {\it in direction} $\alpha$ as
\begin{equation}
\begin{array}{l l l l}
d_{\alpha}f \colon & G & \rightarrow & \textup{GF}(q)\\
& x & \mapsto & f(\alpha + x)\overline{f(x)}\ .
\end{array}
\end{equation} 
\begin{lemma}\label{lem_f_0_f_chap_const}
Let $f \colon G \rightarrow \textup{GF}(q)$. We have
\begin{enumerate}
\item $\forall x \in G^*$, $f(x) = 0$ $\Leftrightarrow$ $\forall \alpha \in G$, $\widehat{f}(\alpha) = f(0_G)$.
\item $\forall \alpha \in G^*$, $\widehat{f}(\alpha) = 0$ $\Leftrightarrow$ $f$ is constant.
\end{enumerate}
\end{lemma}
\begin{proof}
\begin{enumerate}
\item \quad
\begin{itemize}
\item [$\Rightarrow)$] $\widehat{f}(\alpha) = \displaystyle \sum_{x \in G}f(x)\chi_{\alpha}(x) = f(0_G)\chi_{\alpha}(0_G) = f(0_G)$,
\item [$\Leftarrow)$] According to the inversion formula, 
\begin{equation}
\begin{array}{l l l}
f(x) &=& \displaystyle (|G| \bmod p)^{-1}\sum_{\alpha \in G}\widehat{f}(\alpha)\overline{\chi_{\alpha}(x)}\\
&=& \displaystyle f(0_G)(|G| \bmod p)^{-1}\sum_{\alpha \in G}\chi_{-x}(\alpha)\\
&=& 0\  \mbox{for all}\ x \not = 0_G\ .
\end{array}
\end{equation}
\end{itemize}
\item \quad
\begin{itemize}
\item [$\Rightarrow)$] $f(x) = \displaystyle (|G| \bmod p)^{-1}\sum_{\alpha \in G}\widehat{f}(\alpha)\overline{\chi_{\alpha}(x)} = \widehat{f}(0_G)(|G| \bmod p)^{-1}$,
\item [$\Leftarrow)$] $\widehat{f}(\alpha) = \displaystyle \sum_{x \in G}f(x)\chi_{\alpha}(x) = \mathit{constant}\sum_{x \in G}\chi_{\alpha}(x) =0$ for all $\alpha \not = 0_G$. 
\end{itemize}
\end{enumerate}
\end{proof}
\begin{lemma}\label{autocor}
Let $f \colon G \rightarrow \textup{GF}(q)$. We define the autocorrelation function of $f$ as
\begin{equation}
\begin{array}{l l l l}
\mathit{AC}_f \colon & G & \rightarrow & \textup{GF}(q)\\
& \alpha & \mapsto & \displaystyle \sum_{x \in G}d_{\alpha}f(x)\ .
\end{array}
\end{equation}
Then, for all $\alpha \in G$, $\widehat{\mathit{AC}_f}(\alpha) = \textup{norm}(\widehat{f}(\alpha))$.
\end{lemma}
\begin{proof}
Let $\alpha \in G$.
\begin{equation}
\begin{array}{l l l}
\widehat{\mathit{AC}_f}(\alpha) &=& \displaystyle \sum_{x \in G}\mathit{AC}_f (x) \chi_{\alpha}(x)\\
&=& \displaystyle \sum_{x \in G}\sum_{y \in G} d_{x}f(y)\chi_{\alpha}(x)\\
&=& \displaystyle \sum_{x \in G}\sum_{y \in G}f(xy)\overline{f(y)}\chi_{\alpha}(xy)\overline{\chi_{\alpha}(y)}\\
&=& \widehat{f}(\alpha)\overline{\widehat{f}(\alpha)}\\
&=&\textup{norm}(\widehat{f}(\alpha))\ . 
\end{array}
\end{equation}
\end{proof}
We use the above results to obtain the following characterization of bentness as a combinatorial object using the derivative.
\begin{theorem}\label{theorem_derivative_balanced}
The function $f : G \rightarrow \mathcal{S}(\textup{GF}(q))$ is bent if, and only if, for all $\alpha \in G^*$, $\displaystyle \sum_{x \in G}d_{\alpha}f(x) = 0$.
\end{theorem}
\begin{proof}
$\forall \alpha \in G^*,\ \displaystyle \sum_{x \in G}d_{\alpha}f(x) = 0$\\
$\Leftrightarrow$ $\forall \alpha \in G^*,\ \mathit{AC}_f (\alpha) = 0$\\
$\Leftrightarrow$ $\forall \alpha \in G,\ \widehat{\mathit{AC}_f}(\alpha) = \mathit{AC}_f (0_G)$\\
(according to lemma \ref{lem_f_0_f_chap_const})\\
$\Leftrightarrow$ $\forall \alpha \in G,\ \textup{norm}(\widehat{f}(\alpha)) = \displaystyle \sum_{x \in G}f(x)\overline{f(x)}$\\
(according to lemma \ref{autocor})\\
$\Leftrightarrow$ $\forall \alpha \in G,\ \textup{norm}(\widehat{f}(\alpha)) = \displaystyle \sum_{x \in G}\textup{norm}(f(x))$\\
$\Leftrightarrow$ $\forall \alpha \in G,\ \textup{norm}(\widehat{f}(\alpha)) = (|G| \bmod p)$\\
(because $f$ is $\mathcal{S}(\textup{GF}(q))$-valued.) 
\end{proof}

\subsection{Comparison between the two bentness notions}\label{comparison}
In what follows we refer to the traditional bent functions, as introduced in the beginning of section~\ref{sec6}, as ``bent in the usual sense'', while our own bent functions (definition~\ref{new-bent}) are referred to as ``bent in the finite-field sense''. In this subsection we prove that any bent ``well-behaved'' function in the usual sense is also a bent function in the finite-field sense. 

Let $\mathbb{U}_m$ be the group of complex $m$-th roots of unity. Let us assume that $m$ that divides $\sqrt{q}+1$. Therefore, $\mathbb{U}_m$ may be identified with the (unique) sub-group of $\mathcal{S}(\textup{GF}(q))$ of order $m$. We also remark that for every $\omega\in \mathbb{U}_m$, the complex-conjugate $\overline{\omega}=\omega^{-1}$, and if the same $\omega$ is seen as an element of $\mathcal{S}(\textup{GF}(q))$, then also $\omega^{\sqrt{q}}=\omega^{-1}$. Conversely, any sub-group of $\mathcal{S}(\textup{GF}(q))$ may be identified with a sub-group of $\mathbb{U}_{\sqrt{q}+1}$. Let us assume that $G$ belongs to $\mathpzc{Ab}_{\sqrt{q}+1}$. Let us denote by $\widetilde{G}$ the group of complex-valued characters of $G$. We have $\widetilde{G}\cong G\cong \widehat{G}$. It is clear that any complex-valued character of $G$ takes its values in $\mathbb{U}_{\sqrt{q}+1}\cong \mathcal{S}(\textup{GF}(q))$. So that for every $x\in G$, and every $\chi\in \widetilde{G}\cong \widehat{G}$, $f(x)\chi(x)\in \mathbb{U}_{\sqrt{q}+1}\cong \mathcal{S}(\textup{GF}(q))$. Let $\mathbb{Z}[\mathbb{U}_{\sqrt{q}+1}]$ be the group-ring of $\mathbb{U}_{\sqrt{q}+1}$, and let $\overline{\mathbb{U}}_{\sqrt{q}+1}$ be the sub-ring of $\mathbb{C}$ generated by $\mathbb{U}_{\sqrt{q}+1}$. Let $\pi\colon \mathbb{Z}[\mathbb{U}_{\sqrt{q}+1}]\rightarrow \overline{\mathbb{U}}_{\sqrt{q}+1}$ be the unique ring homomorphism such that $\pi([\omega])=\omega$ for all $\omega\in \mathbb{U}_{\sqrt{q}+1}$ (where $[\cdot]\colon \mathbb{U}_{\sqrt{q}+1}\rightarrow \mathbb{Z}[\mathbb{U}_{\sqrt{q}+1}]$ is the canonical inclusion). It is clear that as rings $\overline{\mathbb{U}}_{\sqrt{q}+1}\cong  \mathbb{Z}[\mathbb{U}_{\sqrt{q}+1}]/\ker(\pi)$. Similarly, let $\phi\colon \mathbb{Z}[\mathbb{U}_{\sqrt{q}+1}]\rightarrow \textup{GF}(q)$ be the unique ring homomorphism such that $\phi([\omega])=\tilde{\omega}$ for every $\omega\in \mathbb{U}_{\sqrt{q}+1}$ (where $\tilde{\omega}$ denotes  the image of $\omega$ under an isomorphism $\mathbb{U}_{\sqrt{q}+1}\rightarrow \mathcal{S}(\textup{GF}(q))$). It is easily checked that $\ker(\pi)\subseteq \ker(\phi)$ so that $\phi$ passes to the quotient as a ring homomorphism ${\phi_0}\colon \overline{\mathbb{U}}_{\sqrt{q}+1}\rightarrow \textup{GF}(q)$ such that ${\phi_0}(\omega)=\tilde{\omega}$ for every $\omega\in \mathbb{U}_{\sqrt{q}+1}$ (we observe that the restriction of $\phi_0$ to $\mathbb{U}_{\sqrt{q}+1}$ is precisely the isomorphism $\mathbb{U}_{\sqrt{q}+1}\rightarrow \mathcal{S}(\textup{GF}(q))$ chosen). We notice that for every  integer $n$, $\phi_0(n\omega)=(n\bmod p)\tilde{\omega}$, and $\phi_0(\overline{\omega})=(\tilde{\omega})^{\sqrt{q}}=\overline{\tilde{\omega}}$. Now, let us assume that $f\colon G\rightarrow \mathbb{U}_m$.  Denoting its usual complex-valued Fourier transform by $\tilde{f}$, we have $\phi_0(\tilde{f})=\widehat{f}$. Let us assume that $f$ is bent (in the traditional meaning), \emph{i.e.}, $|\tilde{f}(\alpha)|^2=|G|$ for every $\alpha\in G$. This equivalent to $\tilde{f}(\alpha)\overline{\tilde{f}(\alpha)}=|G|$  for every $\alpha\in G$. Then, $\textup{norm}(\tilde{f}(\alpha))=\phi_0(|G|)=|G|\bmod p$ for every $\alpha\in G$, so that $f$ is bent in this finite-field setting. The following result is then proved.
\begin{theorem}
Let $m$ be a divisor of $\sqrt{q}+1$. Let $G$ be a group in the category $\mathpzc{Ab}_{\sqrt{q}+1}$. Let $f\colon G\rightarrow \mathbb{U}_m$. If $f$ is bent in the usual sense, then it is also bent in the finite-field setting sense. 
\end{theorem}
\noindent This result motivates the study of such bent functions in the finite-field sense.

\subsection{Dual bent function}
Again by analogy to the traditional notion~\cite{Carlet2,KSW}, it is also possible to define a {\it dual} bent function from a given bent function. Actually, as we see it below, $|G|$ must be a square in $\textup{GF}(p)$ to ensure the well-definition of a dual bent. So by using the famous {\it law of quadratic reciprocity}, we can add the following requirement (which contrary to the other assumptions is only needed for proposition \ref{dual_bent}).
\begin{assumption}\label{lrq}
If the prime number $p$ is $\geq 3$, then $|G|$ must also satisfy $|G|^{\frac{p-1}{2}} \equiv 1\pmod p$. If the prime number $p=2$, then there is no other assumptions on $|G|$ (than those already made). 
\end{assumption}
According to assumption~\ref{lrq}, $|G|\bmod p$ is a square in $\textup{GF}(p)$, then there is at least one $x\in \textup{GF}(p)$ with $x^2=|G|\bmod p$. If $p=2$, then $x=1$. If $p\geq 3$, then we choose for $x$ the element $(|G|\bmod p)^{\frac{p+1}{4}}$. Indeed it is a square root of $|G|\bmod p$ since $((|G|\bmod p)^{\frac{p+1}{4}})^2=(|G|\bmod p)^{\frac{p+1}{2}}=(|G|\pmod p)(|G|\pmod p)^{\frac{p-1}{2}}=|G|\pmod p$. In all cases we denote by $(|G|\bmod p)^{\frac{1}{2}}$ the chosen square root of $|G|\bmod p$. Since $|G|\bmod p\not=0$, then it is clear that this square root also is non-zero. Its inverse is denoted by $(|G|\bmod p)^{-\frac{1}{2}}$. Finally it is clear that $((|G|\bmod p)^{-\frac{1}{2}})^2=(|G|\bmod p)^{-1}$.
\begin{proposition}\label{dual_bent}
Let $f \colon G \rightarrow \mathcal{S}(\textup{GF}(q))$ be a bent function, then the following function $\widetilde{f}$, called {\it dual} of $f$, is bent.
\begin{equation}
\begin{array}{l l l l}
\widetilde{f} \colon& G & \rightarrow & \mathcal{S}(\textup{GF}(q))\\
& \alpha & \mapsto & (|G|\bmod p)^{-\frac{1}{2}} \widehat{f}(\alpha)\ .
\end{array}
\end{equation}
\end{proposition}
\begin{proof}
Let us first check that $\widetilde{f}$ is $\mathcal{S}(\textup{GF}(q))$-valued. Let $\alpha \in G$. We have 
\begin{equation}
\begin{array}{l l l}
\widetilde{f}(\alpha)\overline{\widetilde{f}(\alpha)} &=& (|G|\bmod p)^{-\frac{1}{2}} \widehat{f}(\alpha)(|G|\bmod p)^{-\frac{1}{2}} \overline{\widehat{f}(\alpha)}\\
&=& (|G|\bmod p)^{-1}\textup{norm}(\widehat{f}(\alpha))\\
&=&1\ \mbox{(since $f$ is bent.)}
\end{array}
\end{equation}
Let us check that the bentness property holds for $\widetilde{f}$. Let $\alpha \in G$. We have $\widehat{\widetilde{f}}(\alpha) = (|G|\bmod p)^{-\frac{1}{2}}(|G|\bmod p) f(-\alpha)$ (according to formula (\ref{double_FT})). Then
\begin{equation}
\begin{array}{l l l}
\widehat{\widetilde{f}}(\alpha)\overline{\widehat{\widetilde{f}}(\alpha)}&=&(|G|\bmod p) f(-\alpha)\overline{f(-\alpha)}\\
&=& (|G|\bmod p)\textup{norm}(f(-\alpha))\\
&=& (|G|\bmod p)\ \mbox{(since $f$ is $\mathcal{S}(\textup{GF}(q))$-valued.)} 
\end{array}
\end{equation}
\end{proof}
\subsection{Construction of bent functions}
We present a  construction which is actually the translation in our setting of a simple version of the  well-known Maiorana-McFarland construction~\cite{Dill,McF} for classical bent functions.

Let $g \colon G \rightarrow \mathcal{S}(\textup{GF}(q))$ be any function. Let $f$ be the following function.
\begin{equation}
\begin{array}{l l l l}
f\colon & G^2 &\rightarrow & \mathcal{S}(\textup{GF}(q))\\
& (x,y) & \mapsto & \chi_x(y)g(y)\ .
\end{array}
\end{equation}
Then $f$ is bent. We observe that the fact that $f$ is $\mathcal{S}(\textup{GF}(q))$-valued is obvious by construction. So let us prove that $f$ is indeed bent. We use the combinatorial characterization obtained in theorem~\ref{theorem_derivative_balanced}. Let $\alpha,\beta,x,y \in G$. Then we have
\begin{equation}
\begin{array}{l l l}
d_{(\alpha,\beta)}f (x,y)&=&f(\alpha + x,\beta + y)\overline{f(x,y)}\\
&=& \chi_{\alpha + x}(\beta + y)g(\beta + y) \overline{\chi_{x}(y)}\ \overline{g(y)}\\
&=&\chi_{\alpha}(\beta + y)\chi_{x}(\beta + y)g(\beta + y) \overline{\chi_{x}(y)}\ \overline{g(y)}\\
&=&\chi_{\alpha}(\beta)\chi_{\alpha}(y) \chi_{x}(\beta)\chi_x(y)g(\beta + y) \overline{\chi_{x}(y)}\ \overline{g(y)}\\
&=&\chi_{\alpha}(\beta)\chi_{\alpha}(y)g(\beta + y)\overline{g(y)}\chi_{x}(\beta)\\
&=&\chi_{\alpha}(\beta)\chi_{\alpha}(y)g(\beta + y)\overline{g(y)}\chi_{\beta}(x)\ \mbox{(because $\chi_{x}(\beta) = \chi_{\beta}(x)$.)}
\end{array}
\end{equation}
So for $(\alpha,\beta) \in (G^2)^* = G^2 \setminus\{(0_G,0_G)\}$, we obtain
\begin{equation}\label{ega58}
\begin{array}{l l l}
\displaystyle \sum_{(x,y)\in G^2}d_{(\alpha,\beta)}f(x,y) &=& \displaystyle \sum_{(x,y)\in G^2}\chi_{\alpha}(\beta)\chi_{\alpha}(y)g(\beta + y)\overline{g(y)}\chi_{\beta}(x)\\
&=& \displaystyle \chi_{\alpha}(\beta)\sum_{y \in G}\chi_{\alpha}(y)g(\beta + y)\overline{g(y)}\sum_{x \in G}\chi_{\beta}(x)
\end{array}
\end{equation}
The sum $\displaystyle \sum_{x \in G}\chi_{\beta}(x)$ is equal to $0$ if $\beta \not = 0_G$ and $|G| \bmod p$ if $\beta = 0_G$ (according to lemma \ref{lem_sum_car_0_ou_card_G}). Then the right member of the equality (\ref{ega58}) is equal to $0$ if $\beta \not = 0_G$ and $\displaystyle (|G| \bmod p)\chi_{\alpha}(\beta)\sum_{y \in G}\chi_{\alpha}(y)g(\beta + y)\overline{g(y)}$ if $\beta = 0_G$. So when $\beta \not = 0_G$, $\displaystyle \sum_{(x,y)\in G^2}d_{(\alpha,\beta)}f(x,y) = 0$. Now let us assume that $\beta = 0_G$, then because $(\alpha,\beta) \in G^2 \setminus\{(0_G,0_G)\}$, $\alpha \not = 0_G$, we have
\begin{equation}
\begin{array}{l l l}
\displaystyle \sum_{(x,y)\in G^2}d_{(\alpha,0_G)}f(x,y) &=& \displaystyle (|G| \bmod p)\chi_{\alpha}(0_G)\sum_{y \in G}\chi_{\alpha}(y)g(0_G + y)\overline{g(y)}\\
&=&\displaystyle (|G| \bmod p)\sum_{y \in G}\chi_{\alpha}(y)\\
&&\mbox{(because $g$ is $\mathcal{S}(\textup{GF}(q))$-valued)}\\
&=&0\ \mbox{(because $\alpha \not = 0_G$.)}
\end{array}
\end{equation}
So we have checked that for all $(\alpha,\beta) \in G^2 \setminus\{(0_G,0_G)\}$, $\displaystyle \sum_{(x,y)\in G^2}d_{(\alpha,\beta)}f(x,y) = 0$ and then according to theorem \ref{theorem_derivative_balanced} this implies that $f$ is bent.

\section{Vectorial bent functions over a finite field}
\label{sec7}
In this last section is developed a notion of bentness for $\textup{GF}(q)^{l}$-valued functions defined on $G$ called {\it vectorial functions} (this is not the same meaning as in the classical literature where it means in general maps from $\textup{GF}(2)^m$ to $\textup{GF}(2)^n$, see for instance~\cite{Carlet3}). In order to treat this case in a similar way as in  the section~\ref{sec6}, we first introduce a special kind of Fourier transform needed to make clear our definitions.
\subsection{Multidimensional bent functions} 
\begin{definition}
Let $f \colon G \rightarrow \textup{GF}(q)^{l}$. The multidimensional Fourier transform of $f$ is the map $\widehat{f}^{\mathit{MD}}$ defined as
\begin{equation}
\begin{array}{l l l l}
\widehat{f}^{\mathit{MD}}:& G &\rightarrow &\textup{GF}(q)^{l}\\
& \alpha & \mapsto & \displaystyle \sum_{x \in G}\chi_{\alpha}(x)f(x)\ .
\end{array}
\end{equation}
\end{definition}
If $l=1$, then it is obvious that the multidimensional Fourier transform coincides with the classical one. Let $B$ the canonical basis of the $\textup{GF}(q)$-vector space $\textup{GF}(q)^{l}$ of dimension $l$, which is orthonormal for the dot-product $\langle \cdot , \cdot \rangle$ (see formula (\ref{Hdp})). Let $e \in B$. We define the {\it coordinate function} $f_e$ of $f \colon G \rightarrow \textup{GF}(q)^{l}$ with respect to $e$ as
\begin{equation}
\begin{array}{l l l l}
f_e : & G & \rightarrow & \textup{GF}(q)\\
& x & \mapsto & \langle f(x),e\rangle\ .
\end{array}
\end{equation}
Then according to the properties of an orthonormal basis, we observe that
\begin{equation}\label{decomposition_de_phi_suivant_base_de_H}
f(x) = \displaystyle \sum_{e \in B}f_e(x)e
\end{equation}
for each $x \in G$. Thanks to coordinate functions, it is possible to give a connection between the  Fourier transform from section~\ref{sec5} and its multidimensional counterpart.
\begin{lemma}\label{lemma_link_multi_classical_TF}
For each $\alpha \in G$, we have
\begin{equation}
\widehat{f}^{\mathit{MD}}(\alpha) = \displaystyle \sum_{e \in B}\widehat{f_e}(\alpha)e\ .
\end{equation}
\end{lemma}
\begin{proof}
Let $\alpha \in G$.
\begin{equation}
\begin{array}{l l l}
\widehat{f}^{\mathit{MD}}(\alpha)&=&\displaystyle \sum_{x \in G}\chi_{\alpha}(x)f(x)\\
&=& \displaystyle \sum_{x \in G}\sum_{e \in B}\chi_{\alpha}(x)f_e(x)e\\
&=&\displaystyle \sum_{e \in B}\left ( \sum_{x\in G}f_e (x) \chi_{\alpha}(x)\right )e\\
&=&\displaystyle \sum_{e\in B}\widehat{f_e}(\alpha)e\ .  
\end{array}
\end{equation}
\end{proof}
Hereafter in this subsection are established some properties for the multidimensional Fourier transform similar to the corresponding properties of the ``one-dimensional" Fourier transform. So let $f : G \rightarrow \textup{GF}(q)^{l}$.
Let us compute the Fourier transform of $\widehat{f}^{\mathit{MD}}$. Let $\alpha \in G$.
\begin{equation}
\begin{array}{l l l}
\widehat{\widehat{f}^{\mathit{MD}}}^{\mathit{MD}} (\alpha) &=& \displaystyle \sum_{x \in G}\chi_{\alpha} (x) \widehat{f}^{\mathit{MD}}(x)\\
& = & \displaystyle \sum_{x \in G} \sum_{e \in B} \widehat{f_e}(x) \chi_{\alpha} (x) e\ \mbox{(according to lemma \ref{lemma_link_multi_classical_TF})}\\
& = & \displaystyle \sum_{E \in B} \left( \sum_{x \in G} \widehat{f_{e}} (x) \chi_{\alpha} (x)\right) e\\
& = & \displaystyle \sum_{e \in B} \widehat{\widehat{f_e}} (\alpha) e\\
& = & \displaystyle  (|G| \bmod p)\sum_{e \in B} f_e (- \alpha) e\ \mbox{(according to relation (\ref{double_FT}))}\\
& = & (|G| \bmod p) f(-\alpha)\ \mbox{(according to formula (\ref{decomposition_de_phi_suivant_base_de_H}))}\ .
\end{array}
\end{equation}
The equality $\widehat{\widehat{f}^{\mathit{MD}}}^{\mathit{MD}} (\alpha) =  (|G| \bmod p)f(-\alpha)$ will be useful in the sequel. Moreover the following {\it inversion formula} is proved.
\begin{equation}\label{inversion_formula_for_multi_dim_TF}
\mbox{For all}\ \alpha \in G,\ f(\alpha) = \displaystyle  (|G| \bmod p)^{-1} \sum_{x \in G} \overline{\chi_{x} (\alpha)} \widehat{f}^{\mathit{MD}}(x)\ .
\end{equation}
Now, we present a certain kind of Parseval equation in this context.
\begin{theorem}[Parseval equation]
Let $f \colon G \rightarrow \textup{GF}(q)^l$ then 
\begin{equation}
\displaystyle \sum_{x \in G}\textup{norm}_l(f(x)) =  (|G| \bmod p)^{-1} \sum_{\alpha \in G} \textup{norm}_l(\widehat{f}^{\mathit{MD}}(\alpha))\ .
\end{equation}
If $f \colon G \rightarrow \mathcal{S}(\textup{GF}(q)^l)$, then
\begin{equation}
\displaystyle \sum_{\alpha \in G} \textup{norm}_l(\widehat{f}^{\mathit{MD}}(\alpha)) = (|G| \bmod p)^2\ .
\end{equation}
\end{theorem}
\begin{proof}
\begin{equation}
\begin{array}{l l l}
\displaystyle \sum_{x \in G} \textup{norm}_l(f(x)) &=& \displaystyle \sum_{x \in G} \sum_{e \in B} \textup{norm}(f_e(x))\\
&=&\displaystyle (|G| \bmod p)^{-1}\sum_{e \in B} \sum_{\alpha \in G} \textup{norm}(\widehat{f_e}(\alpha))\\
&&\mbox{(according to the Parseval equation applied on $f_e$)}\\
&=&\displaystyle (|G| \bmod p)^{-1} \sum_{\alpha \in G} \sum_{e \in B} \textup{norm}(\widehat{f_e}(\alpha))\\
&=&\displaystyle (|G| \bmod p)^{-1} \sum_{\alpha \in G} \textup{norm}_l(\widehat{f}^{\mathit{MD}}(\alpha))\ .
\end{array}
\end{equation}
The second assertion is obvious. 
\end{proof}
It is possible, and even more interesting, to obtain this Parseval equation in an alternative way. Let $f,g\in (\textup{GF}(q)^l)^{G}$ and $\alpha \in G$. By replacing the multiplication by the dot-product, we define the {\it convolutional product} as follows
\begin{equation}
(f * g)(\alpha) = \displaystyle \sum_{x \in G} \langle g(\alpha + x),f(x)\rangle\ .
\end{equation}
Since $f * g \colon G \rightarrow \textup{GF}(q)$, we can compute its one-dimensional Fourier transform
\begin{equation}\label{trivialiZation}
\begin{array}{l l l}
\widehat{(f * g)}(\alpha) &=& \displaystyle \sum_{x \in G}(f * g)(x) \chi_{\alpha}(x)\\
&=& \displaystyle \sum_{x \in G} \sum_{y \in G} \chi_{\alpha} (x) \langle g(x+y),f(y)\rangle\\
&=&\displaystyle \sum_{x \in G}\sum_{y \in G} \chi_{\alpha}(x+y)\overline{\chi_{\alpha}(y)}\langle g(x+y),f(y)\rangle\\
&=&\displaystyle \sum_{x \in G}\sum_{y \in G} \langle \chi{\alpha}(x+y) g(x+y),\chi_{\alpha}(y)f(y)\rangle\\
&=&\displaystyle \sum_{y \in G}\langle \sum_{x \in G} \chi_{\alpha} (x+y) g(x+y),\chi_{\alpha} (y) f(y)\rangle\\ 
&=&\displaystyle \sum_{y \in G} \langle \widehat{g}^{\mathit{MD}}(\alpha),\chi_{\alpha} (y) f(y)\rangle\\
&=&\displaystyle \langle \widehat{g}^{\mathit{MD}}(\alpha),\sum_{y \in G} \chi_{\alpha} (y) f(y)\rangle\\
&=&\langle \widehat{g}^{\mathit{MD}}(\alpha),\widehat{f}^{\mathit{MD}}(\alpha)\rangle\ .
\end{array}
\end{equation}
It is a kind of trivialization of the convolutional product by the Fourier transform. Now let us compute $(f * g)(0_G)$. There are two ways to do this. The first one is given by definition: $(f * g)(0_G) = \displaystyle \sum_{x \in G} \langle g(x),f(x)\rangle$. The second one is given by the inversion formula of the Fourier transform.
\begin{equation}
\begin{array}{l l l}
(f * g)(0_G) &=& \displaystyle (|G|\bmod p)^{-1} \sum_{\alpha \in G}\widehat{(f * g)}(\alpha) \overline{\chi_{0_G}(\alpha)}\\
&=&\displaystyle (|G|\bmod p)^{-1} \sum_{\alpha \in G}\widehat{(f * g)}(\alpha)\\
&=& \displaystyle (|G|\bmod p)^{-1}  \sum_{\alpha \in G} \langle \widehat{g}^{\mathit{MD}}(\alpha),\widehat{f}^{\mathit{MD}}(\alpha)\rangle\ .
\end{array}
\end{equation}
Then we have $\displaystyle \sum_{x \in G} \langle g(x),f(x) \rangle = (|G|\bmod p)^{-1} \sum_{\alpha \in G}\langle \widehat{g}^{\mathit{MD}}(\alpha),\widehat{f}^{\mathit{MD}}(\alpha)\rangle$.\\
Now let $f = g$, then
\begin{equation}
\displaystyle \sum_{x \in G} \langle f(x), f(x)\rangle = (|G|\bmod p)^{-1} \sum_{\alpha \in G} \langle \widehat{f}(\alpha),\widehat{f}(\alpha)\rangle
\end{equation}
{\it i.e.},
\begin{equation}
\displaystyle \sum_{x \in G} \textup{norm}_l(f(x)) = (|G|\bmod p)^{-1} \sum_{\alpha \in G} \textup{norm}_l(\widehat{f}(\alpha))\ .
\end{equation}

\subsection{Multidimensional bent functions}
In the paper~\cite{Poi05} is introduced the notion of multidimensional bentness for $\mathcal{H}$-valued functions defined on a finite  Abelian group $G$, where $\mathcal{H}$ is a finite-dimensional Hermitian space. In this subsection, we translate this notion in our special kind of Hermitian structure. 
\begin{definition}
Let $f : G \rightarrow \mathcal{S}(\textup{GF}(q)^l)$. The function $f$ is said {\it multidimensional bent} if for all $\alpha \in G$, $\textup{norm}_l(\widehat{f}^{\mathit{MD}}(\alpha)) = (|G| \bmod p)$.
\end{definition}
\begin{lemma}\label{lemma1}
Let $f \colon {G} \rightarrow \textup{GF}(q)^l$. Then, $f(x) = (\underbrace{0,\ldots,0}_{l\ \mathit{times}})$ for all $x \in G^*$ if, and only if, $\widehat{f}^{\mathit{MD}}(\alpha) = f(0_G)$ for all $\alpha \in G$.
\end{lemma}
\begin{proof} $\Rightarrow$) $\widehat{f}^{\mathit{MD}}(\alpha) = \displaystyle \sum_{x \in {G}} \chi_{\alpha}(x) f(x) = f(0_G)$ $\forall \alpha \in {G}$.\\
$\Leftarrow$) $f(x) = \displaystyle (|G|\bmod p)^{-1} \sum_{\alpha \in {G}} \overline{\chi_{\alpha} (x)} \widehat{f}^{\mathit{MD}}(\alpha)$ (by the inversion formula of the multidimensional Fourier transform). Then by assumption, $$f(x) = \displaystyle (|G|\bmod p)^{-1} f(0_G) \sum_{\alpha \in {G}} \chi_{\alpha} (x) =(\underbrace{0,\ldots,0}_{l\ \mathit{times}})$$ if $x \in G^*$, and $f(0_G)$ otherwise.  
\end{proof}
This technical result holds in particular when $l=1$ which is lemma \ref{lem_f_0_f_chap_const}.

As in the one-dimensional setting, there exists a combinatorial characterization of the multidimensional bentness. We define a kind of derivative for $\textup{GF}(q)^l$-valued functions. Another time we use the natural ``multiplication'' of $\textup{GF}(q)^l$ which is its dot-product.
\begin{definition}
Let $f \colon {G}  \rightarrow \textup{GF}(q)^l$ and $\alpha \in {G}$. The {\it derivative of $f$ in direction $\alpha$} is defined by
\begin{equation}
\begin{array}{l l l l}
d_{\alpha}f \colon & {G} & \rightarrow & \textup{GF}(q)\\
& x & \mapsto & \langle f(\alpha + x), f(x)\rangle\ .
\end{array}
\end{equation}
\end{definition}
This derivative measures the default of orthogonality between $f(x)$ and $f(\alpha + x)$.
\begin{proposition}
Let $f \colon {G} \rightarrow \mathcal{S}(\textup{GF}(q)^l)$. Then, $f$ is bent if, and only if, for all $\alpha \in G^*$, $\widehat{d_{\alpha}f}(0_G) = 0$.
\end{proposition}
\begin{proof}
Let us define the following autocorrelation function
\begin{equation}
\begin{array}{l l l l}
\mathit{AC}_{f} \colon & {G} & \rightarrow & \textup{GF}(q)\\
& \alpha & \mapsto & \widehat{d_{\alpha}f}(0_G)\ .
\end{array}
\end{equation}
We have
\begin{equation}
\begin{array}{l l l}
\widehat{d_{\alpha}f}(0_G) &=& \displaystyle \sum_{x \in G}d_{\alpha}f(x) \chi_{0_G}(x)\\
&=&\displaystyle \sum_{x \in {G}}d_{\alpha}f(x)\\
&=&\displaystyle \sum_{x \in {G}}\langle f(\alpha + x),f(x)\rangle\\
&=& (f * f)(\alpha)\ .
\end{array}
\end{equation}
Let us compute $\widehat{\mathit{AC}_{f}}(\alpha)$.
\begin{equation}
\begin{array}{l l l}
\widehat{\mathit{AC}_{f}}(\alpha) &=& \displaystyle \sum_{x \in {G}} \mathit{AC}_{f}(x)\chi_{\alpha} (x)\\
&=&\displaystyle \sum_{x \in {G}} (f * f)(x)\chi_{\alpha} (x)\\
&=&\widehat{(f * f)}(\alpha)\\
&=&\langle \widehat{f}^{\mathit{MD}}(\alpha),\widehat{f}^{\mathit{MD}}(\alpha)\rangle\ \mbox{(by the formula (\ref{trivialiZation}))}\\
&=&\textup{norm}_l(\widehat{f}^{\mathit{MD}}(\alpha))\ .
\end{array}
\end{equation}
Then we have\\
$\forall \alpha \in G^*$, $\widehat{d_{\alpha}f}(0_G) = 0$\\
$\Leftrightarrow$ $\forall \alpha \in G^*$, $\mathit{AC}_{f}(\alpha) = 0$\\
$\Leftrightarrow$ $\forall \alpha \in {G}$, $\widehat{\mathit{AC}_{f}}(\alpha) = \mathit{AC}_{f}(0_G)$ \mbox{(according to lemma \ref{lemma1})}\\
$\Leftrightarrow$ $\forall \alpha \in G$, $\textup{norm}_l(\widehat{f}^{\mathit{MD}}(\alpha)) = \mathit{AC}_{f}(0_G)$.\\
As $\mathit{AC}_{f}(0_G) = (f * f)(0_G) = \displaystyle \sum_{x \in {G}} \langle f(x),f(x)\rangle = \sum_{x \in {G}} \textup{norm}(f(x)) = (|G| \bmod p)$ (since $f$ is $\mathcal{S}(\textup{GF}(q)^l)$-valued), we conclude with the expected result. 
\end{proof}



\bibliographystyle{elsarticle-num}

\end{document}